\documentclass[conference,letter,romanappendices,draft,onecolumn]{ieeeconf}
\let\proof\relax   

\usepackage{amsthm,xpatch}
\usepackage{amsmath,amsfonts}
\usepackage{cite}
\usepackage{amssymb}
\usepackage{dsfont}
\usepackage{graphicx, subfigure}
\usepackage{color}
\usepackage{breqn}
\usepackage{mathtools}
\usepackage{bbm}
\usepackage{latexsym}
\usepackage[ruled, linesnumbered]{algorithm2e}
\usepackage{accents}
\usepackage{tikz}

\newtheorem{lemma}{Lemma}
\newtheorem{theorem}{Theorem}
\newtheorem{remark}{Remark}

\newtheorem{example}{Example}

\newcommand*{\transpose}{%
  {\mathpalette\@transpose{}}%
}

\IEEEoverridecommandlockouts

\begin{document}

\newcommand{\SB}[3]{
\sum_{#2 \in #1}\biggl|\overline{X}_{#2}\biggr| #3
\biggl|\bigcap_{#2 \notin #1}\overline{X}_{#2}\biggr|
}

\newcommand{\Mod}[1]{\ (\textup{mod}\ #1)}

\newcommand{\overbar}[1]{\mkern 0mu\overline{\mkern-0mu#1\mkern-8.5mu}\mkern 6mu}

\makeatletter
\newcommand*\nss[3]{%
  \begingroup
  \setbox0\hbox{$\m@th\scriptstyle\cramped{#2}$}%
  \setbox2\hbox{$\m@th\scriptstyle#3$}%
  \dimen@=\fontdimen8\textfont3
  \multiply\dimen@ by 4             
  \advance \dimen@ by \ht0
  \advance \dimen@ by -\fontdimen17\textfont2
  \@tempdima=\fontdimen5\textfont2  
  \multiply\@tempdima by 4
  \divide  \@tempdima by 5          
  \ifdim\dimen@<\@tempdima
    \ht0=0pt                        
    \@tempdima=\fontdimen5\textfont2
    \divide\@tempdima by 4          
    \advance \dimen@ by -\@tempdima 
    \ifdim\dimen@>0pt
      \@tempdima=\dp2
      \advance\@tempdima by \dimen@
      \dp2=\@tempdima
    \fi
  \fi
  #1_{\box0}^{\box2}%
  \endgroup
  }
\makeatother

\makeatletter
\renewenvironment{proof}[1][\proofname]{\par
  \pushQED{\qed}%
  \normalfont \topsep6\p@\@plus6\p@\relax
  \trivlist
  \item[\hskip\labelsep
        \itshape
    #1\@addpunct{:}]\ignorespaces
}{%
  \popQED\endtrivlist\@endpefalse
}
\makeatother

\makeatletter
\newsavebox\myboxA
\newsavebox\myboxB
\newlength\mylenA

\newcommand*\xoverline[2][0.75]{%
    \sbox{\myboxA}{$\m@th#2$}%
    \setbox\myboxB\null
    \ht\myboxB=\ht\myboxA%
    \dp\myboxB=\dp\myboxA%
    \wd\myboxB=#1\wd\myboxA
    \sbox\myboxB{$\m@th\overline{\copy\myboxB}$}
    \setlength\mylenA{\the\wd\myboxA}
    \addtolength\mylenA{-\the\wd\myboxB}%
    \ifdim\wd\myboxB<\wd\myboxA%
       \rlap{\hskip 0.5\mylenA\usebox\myboxB}{\usebox\myboxA}%
    \else
        \hskip -0.5\mylenA\rlap{\usebox\myboxA}{\hskip 0.5\mylenA\usebox\myboxB}%
    \fi}
\makeatother

\xpatchcmd{\proof}{\hskip\labelsep}{\hskip3.75\labelsep}{}{}

\pagestyle{empty}

\title{\fontsize{21}{28}\selectfont Private Computation with Individual and Joint Privacy}

\author{Anoosheh Heidarzadeh and Alex Sprintson\thanks{The authors are with the Department of Electrical and Computer Engineering, Texas A\&M University, College Station, TX 77843 USA (E-mail: \{anoosheh, spalex\}@tamu.edu).}}

\maketitle 

\thispagestyle{empty}

\begin{abstract}
This paper considers the problem of single-server Private Computation (PC) in the presence of Side Information (SI). In this problem, there is a server that stores $K$ i.i.d.~messages, and a user who has a subset of $M$ uncoded messages or a coded linear combination of them as side information, where the identities of these messages are unknown to the server. The user wants to privately compute (via downloading information from the server) a linear combination of a subset of $D$ other messages, where the identities of these messages must be kept private individually or jointly. For each setting, we define the capacity as the supremum of all achievable download rates. 

We characterize the capacity of both PC with coded and uncoded SI when individual privacy is required, for all $K,M,D$. Our results indicate that both settings have the same capacity. 
In addition, we establish a non-trivial lower bound on the capacity of PC with coded SI when joint privacy is required, for a range of parameters $K,M,D$. This lower bound is the same as the lower bound we previously established on the capacity of PC with uncoded SI when joint privacy is required. 
\end{abstract}

\section{introduction}
In this work, we consider the problem of Private Computation (PC) in the presence of side information. In this problem, there is a single (or multiple) remote server(s) storing (identical copies of) a database of i.i.d.~messages; and there is a user who initially has a \emph{side information} about some subset of messages in the database, where the identities of the messages in the support of the side information are initially unknown to the server. The user is interested in privately computing (via downloading information from the server(s)) a linear combination of a different subset of database messages, while minimizing the total amount of information being downloaded from the server(s). 

We consider two different types of side information: (i) \emph{uncoded side information} - where the user knows a subset of database messages, and (ii) \emph{coded side information} - where the user holds a linear combination of a subset of database messages. These settings are referred to as \emph{PC with Side Information (PC-SI)} and \emph{PC with Coded Side Information (PC-CSI)}, respectively. We also consider two different privacy conditions: (i) \emph{individual privacy} - where the identity of each message in the support set of the demanded linear combination needs to be kept private individually, and (ii) \emph{joint privacy} - where the identities of all messages in the support set of the demanded linear combination must be kept private jointly. When the condition~(i) or (ii) is required, we refer to the PC problem as \emph{Individually-Private Computation (IPC)} or \emph{Jointly-Private Computation (JPC)}, respectively. 

The goal is to design a protocol for generating the query of the user and the corresponding answer of the server(s) such that the entropy of the answer is minimized, while the query satisfies the privacy condition.      

Both IPC and JPC settings are related to the problem of Private Computation, introduced in~\cite{SJ2018}, where the goal is to compute a linear combination of the messages in the database, while hiding both the identities and the coefficients of these messages. Several variants of this problem were also studied in~\cite{OK2018,OLRK2018,MM2018,CWJ2018}. These works consider neither individual nor joint privacy, nor any type of side information. 

The JPC-SI setting, initially introduced in~\cite{HS2019PC}, is closely related to the problem of Private Information Retrieval with Side Information (PIR-SI), which was initially introduced in~\cite{Kadhe2017,KGHERS2017} and later extended in several works, e.g.,~\cite{HKGRS:2018,LG:2018,KKHS32019,KHSO2019}. In the PIR-SI problem, a user wishes to retrieve a subset of database messages with the help of an uncoded side information, while achieving joint privacy. Several variants of PIR with different types of side information or different types of privacy conditions were also studied in~\cite{Chen2017side,Maddah2018,Tandon2017,Wei2017CacheAidedPI,Wei2017FundamentalLO,WU2018,HKS2019,KKHS22019}. 
The IPC-SI setting is an extension of the PIR-SI problem when individual privacy is required. This problem, known as IPIR-SI, was introduced in~\cite{HKRS2019}.  
The JPC-CSI ad IPC-CSI settings are two generalizations of PIR with Coded Side Information (PIR-CSI), previously studied in~\cite{HKS2018} and~\cite{KKHS12019}. 

\subsection{Main Contributions}
In this work, we focus on the single-server case. 
For each type of side information (coded or uncoded) and each privacy condition (individual or joint), the \emph{capacity} of the underlying setting is defined as the supremum of all achievable download rates, where the \emph{download rate} is the ratio of the entropy of a message to the entropy of the server's answer. 

We characterize the capacity of both the IPC-SI and IPC-CSI settings, for all parameters. These results subsume several existing results in the PIR literature. 
The converse proof is information-theoretic, and the achievability scheme is a generalization of our recently proposed scheme in~\cite{HKS2019Journal} for the PIR-CSI setting. Our results show that the capacity of both settings are the same. This implies that, when individual privacy is required, having only \emph{one} linear combination of a subset of messages as side information is as efficient as having them all separately. In addition, we establish a non-trivial lower bound on the capacity of the JPC-CSI setting for a range of parameters. Interestingly, this lower bound is the same as the lower bound we previously established in~\cite{HS2019PC} on the capacity of the JPC-SI setting. The proof of achievability is based on a modification of the scheme we proposed in~\cite{HS2019PC} for the JPC-SI setting.   

Our results for both IPC and JPC settings, when compared to the existing results in the PIR literature, indicate that one can privately compute a linear combination of multiple messages much more efficiently than privately retrieving multiple messages, and linearly  combining them locally. In addition, comparing our results with those in~\cite{SJ2018}, one can see that hiding only the identities of the messages (either individually or jointly) and not their coefficients ---which may still provide a satisfactory level of privacy in many applications, can be done much less costly, even when there is only one server and/or the user has no side information.      

\section{Problem Formulation}\label{sec:SN}
Throughout, random variables and their realizations are denoted by bold-face letters and regular letters, respectively.  

Let $\mathbb{F}_q$ be a finite field for a prime $q$, and let $\mathbb{F}_{q^{\ell}}$ be an extension field of $\mathbb{F}_q$ for a positive integer $\ell$. 

Let $K$, $M$, and $D$ be non-negative integers such that ${K\geq M+D}$. Let $\mathcal{K} \triangleq \{1,\dots,K\}$, and let $\mathcal{K}_M$ (or $\mathcal{K}_D$) be the set of all $M$-subsets (or $D$-subsets) of $\mathcal{K}$. Let $\mathcal{C}$ be the set of all nonzero elements in $\mathbb{F}_q$, and let $\mathcal{C}_{M}$ (or $\mathcal{C}_{D}$) be the set of all ordered multisets of $\mathcal{C}$ of size $M$ (or $D$).

Consider a single server that stores a dataset of $K$ messages,  $X_{\mathcal{K}}\triangleq \{X_1,\dots,X_K\}$, where each message $X_i$ is independently and uniformly distributed over $\mathbb{F}_{q^{\ell}}$. That is, ${H(\mathbf{X}_i) = L}$ for ${i\in \mathcal{K}}$, and $H(\mathbf{X}_{\mathcal{K}}) = KL$, where $\mathbf{X}_{\mathcal{K}}\triangleq \{\mathbf{X}_1,\dots,\mathbf{X}_K\}$, and $L \triangleq \ell\log_2 q$. Consider a user that knows a linear combination $Y^{[S,U]}\triangleq \sum_{i\in S} u_i X_i$ of $M$ messages $X_S\triangleq \{X_{i}\}_{i\in S}$ for some ${S\in\mathcal{K}_M}$ and some $U\triangleq \{u_i\}_{i\in S}\in \mathcal{C}_M$, and wishes to retrieve a linear combination $Z^{[W,V]}\triangleq \sum_{i\in W} v_i X_i$ from the server for some $W\in \mathcal{K}_D$, $W\cap S= \emptyset$, and some ${V \triangleq \{v_{i}\}_{i\in W} \in \mathcal{C}_D}$. We refer to $Y^{[S,U]}$ as the \emph{side information}, $X_S$ as the \emph{side information support set}, $S$ as the \emph{side information support index set}, $M$ as the \emph{side information support size}, $Z^{[W,V]}$ as the \emph{demand}, $X_W$ as the \emph{demand support set}, $W$ as the \emph{demand support index set}, and $D$ as the \emph{demand support size}.  

We assume that $\mathbf{S}$, $\mathbf{U}$, and $\mathbf{V}$ are distributed uniformly over $\mathcal{K}_M$, $\mathcal{C}_M$, and $\mathcal{C}_D$, respectively, and $\mathbf{W}$, given $\mathbf{S}=S$, is uniformly distributed over all $W\in \mathcal{K}_D$, $W\cap S=\emptyset$. 
Also, we assume that the server initially knows $M,D$, and the joint distribution of $\mathbf{S}$, $\mathbf{U}$, $\mathbf{W}$, and $\mathbf{V}$, whereas the realizations $S$, $U$, $W$ and $V$ are not initially known to the server. 

For any given $S,U,W,V$, the user sends to the server a query $Q^{[W,V,S,U]}$, which is a (potentially stochastic) function of $W$, $V$, $S$, and $U$, in order to retrieve $Z^{[W,V]}$. For simplifying the notation, we denote $\mathbf{Q}^{[\mathbf{W},\mathbf{V},\mathbf{S},\mathbf{U}]}$ by $\mathbf{Q}$. The query must satisfy one of the following two privacy conditions: 

\begin{itemize}
\item[(i)] \emph{Individual Privacy:} every message in $X_{\mathcal{K}}$ must be equally likely to be in the user's demand support set, i.e., for all $i\in \mathcal{K}$, it must hold that \[\Pr(i\in \mathbf{W}|\mathbf{Q} = Q^{[W,V,S,U]}) = \Pr(i\in \mathbf{W}).\] \item[(ii)] \emph{Joint Privacy:} every $D$-subset of messages in $X_{\mathcal{K}}$ must be equally likely to be the user's demand support set, i.e., for all $W^{*}\in \mathcal{K}_D$, it must hold that \[\Pr(\mathbf{W} = W^{*}|\mathbf{Q} = Q^{[W,V,S,U]}) = \Pr(\mathbf{W} = W^{*}).\] 	
\end{itemize}

The joint privacy, which is a stronger notion of privacy, implies the individual privacy, but not vice versa. The main difference between these two conditions is that for joint privacy the query must protect the correlation between the indices in the demand support index set, whereas for individual privacy some information about this correlation may be leaked, and hence a weaker notion of privacy.  

Neither individual nor joint privacy requires the privacy of the coefficients in the demand to be protected. This is in contrast to the privacy condition being considered in~\cite{SJ2018}, and as a result of this relaxation one can expect more efficient private computation schemes in our settings. In particular, for single-server private computation without any side information, the user must download the entire dataset in order to protect the privacy of both the identities of the messages in the demand support set and their coefficients in the demand~\cite{SJ2018}. However, for neither of the two privacy conditions being considered here the entire dataset needs to be downloaded, even when the user has no side information.  

Upon receiving $Q^{[W,V,S,U]}$, the server sends to the user an answer $A^{[W,V,S,U]}$, which is a (deterministic) function of the query $Q^{[W,V,S,U]}$ and the messages in $X_{\mathcal{K}}$. We denote $\mathbf{A}^{[\mathbf{W},\mathbf{V},\mathbf{S},\mathbf{U}]}$ by $\mathbf{A}$ for the ease of notation. Note that $H(\mathbf{A}| \mathbf{Q},\mathbf{X}_{\mathcal{K}}) = 0$, since $(\mathbf{W},\mathbf{V},\mathbf{S},\mathbf{U})$ and $\mathbf{A}$ are conditionally independent given $(\mathbf{Q},\mathbf{X}_{\mathcal{K}})$. 

The collection of $A^{[W,V,S,U]}$, $Q^{[W,V,S,U]}$, $Y^{[S,U]}$, $W$, $V$, $S$, and $U$ must enable the user to retrieve the demand $Z^{[W,V]}$. That is, it must hold that \[H(\mathbf{Z}^{[\mathbf{W},\mathbf{V}]}| \mathbf{A}, \mathbf{Q}, \mathbf{Y}^{[\mathbf{S},\mathbf{U}]},\mathbf{W},\mathbf{V},\mathbf{S},\mathbf{U})=0.\] We refer to this condition as the \emph{recoverability condition}. 

For each type of privacy, the problem is to design a protocol for generating a query ${Q}^{[W,V,S,U]}$ (and the corresponding answer ${A}^{[W,V,S,U]}$, given ${Q}^{[W,V,S,U]}$ and ${X}_{\mathcal{K}}$) for any given $W,V,S,U$, such that both the privacy and recoverability conditions are satisfied. We refer to this problem as \emph{single-server} \emph{Individually-Private Computation with Coded Side Information (IPC-CSI)} or \emph{Jointly-Private Computation with Coded Side Information (JPC-CSI)}, when individual or joint privacy is required, respectively. 

We similarly define the \emph{IPC-SI} and \emph{JPC-SI} problems for the settings in which the user's side information is the support set $X_S$ itself, instead of a linear combination of the messages in $X_S$. 

We refer to a protocol that generates query/answer for the IPC-CSI or JPC-CSI setting as an \emph{IPC-CSI} or a \emph{JPC-CSI} \emph{protocol}, respectively. The \emph{rate} of an IPC-CSI or a JPC-CSI protocol is defined as the ratio of the entropy of a message, i.e., $L$, to the entropy of the answer $\mathbf{A}$. The \emph{capacity} of the IPC-CSI or JPC-CSI setting is defined as the supremum of rates over all IPC-CSI or JPC-CSI protocols, respectively. An IPC-SI or a JPC-SI protocol, its rate,  and the capacity of the IPC-SI or JPC-SI setting are defined similarly. 

Our goal in this work is to establish lower and/or upper bounds on the capacity of IPC-CSI, JPC-CSI, IPC-SI, and JPC-SI settings, in terms of the parameters $K$, $M$, and $D$.

\section{Main Results}
Our main results for the IPC and JPC settings with both coded and uncoded side information are summarized in Sections~\ref{subsec:IPC} and~\ref{subsec:JPC}, respectively. 

The following two lemmas provide a necessary condition for individual and joint privacy, for both types of side information. The proofs are straightforward by the way of contradiction, and hence omitted for brevity. 

\begin{lemma}[A Necessary Condition for Individual Privacy]\label{lem:1}
For any $i\in\mathcal{K}$, there exist $W^{*}\in\mathcal{K}_D$, ${V^{*}\in\mathcal{C}_D}$, and $S^{*}\in\mathcal{K}_{M}$ where $i\in W^{*}$ and $S^{*}\cap W^{*} = \emptyset$, such that \[H(\mathbf{Z}^{[W^{*},V^{*}]}| \mathbf{A}, \mathbf{Q}, \mathbf{X}_{S^{*}}) = 0.\] 	
\end{lemma}

\begin{lemma}[A Necessary Condition for Joint Privacy]\label{lem:2}
For any $W^{*}\in\mathcal{K}_D$, there exist $V^{*}\in\mathcal{C}_D$ and $S^{*}\in\mathcal{K}_{M}$ where ${S^{*}\cap W^{*} = \emptyset}$, such that \[H(\mathbf{Z}^{[W^{*},V^{*}]}| \mathbf{A}, \mathbf{Q}, \mathbf{X}_{S^{*}}) = 0.\] 		
\end{lemma}

Thinking of scalar-linear IPC or JPC protocols ---where the answer consists only of scalar-linear combinations of the messages in $X_{\mathcal{K}}$, the necessary conditions in Lemmas~\ref{lem:1} and~\ref{lem:2} imply the need for linear codes that satisfy certain combinatorial requirements. (Recently, in~\cite{KHSO2019}, we made a similar connection between single-server PIR with side information and locally recoverable codes.) Consider a (linear) code of length $K$ that satisfies the following requirement: for any $i\in \{1,\dots,K\}$, there is a codeword of (Hamming) weight $D$ or $M+D$ (or at least $D$ and at most $M+D$) whose support includes the index $i$. The parity-check equations of the dual of any such code can be used for constructing a scalar-linear IPC-CSI (or IPC-SI) protocol. Minimizing the entropy of the answer in order to maximize the rate of the protocol translates into minimizing the dimension of the code. In this work, we design optimal codes with minimum dimension for all $K,M,D$ for the IPC-CSI setting. These codes naturally serve also as optimal codes for the IPC-SI setting. 

The problem of designing a scalar-linear JPC-CSI (or JPC-SI) protocol reduces to the problem of designing a code of length $K$ with minimum dimension satisfying the following requirement: for any $D$-subset $W\subseteq \{1,\dots,K\}$, there is a codeword of weight $D$ or $M+D$ (or at least $D$ and at most $M+D$) whose support includes the $D$-subset $W$. The design of optimal codes satisfying this requirement remains an open problem. In~\cite{HS2019PC}, we initiated the study of the JPC-SI setting, and established a non-trivial upper bound on the dimension of optimal codes for this setting. In this work, we make the first attempt towards characterizing the dimension of optimal codes for the JPC-CSI setting; and provide a non-trivial upper bound for a range of parameters $K,M,D$.

\subsection{IPC-SI and IPC-CSI}\label{subsec:IPC}
The capacity of IPC-SI and IPC-CSI for arbitrary $K,M,D$ are characterized in Theorems~\ref{thm:IPC-SI} and~\ref{thm:IPC-CSI}, respectively. 

\begin{theorem}\label{thm:IPC-SI}
For the IPC-SI setting with $K$ messages, side information of size $M$, and demand support size $D$, the capacity is given by ${\lceil \frac{K}{M+D} \rceil}^{-1}$. 	
\end{theorem}

\begin{theorem}\label{thm:IPC-CSI}
For the IPC-CSI setting with $K$ messages, side information support size $M$, and demand support size $D$, the capacity is given by ${\lceil \frac{K}{M+D} \rceil}^{-1}$. 	
\end{theorem}

For the converse proof, we use information-theoretic arguments relying primarily on the result of Lemma~\ref{lem:1}, to upper bound the rate of any IPC-SI protocol (see Section~\ref{subsec:IPC-Conv}). This upper bound obviously holds for any IPC-CSI protocol. For the proof of achievability, we construct a new scalar-linear IPC-CSI protocol, termed the \emph{Generalized Modified Partition-and-Code (GMPC) protocol}, which achieves the rate upper bound (see Section~\ref{subsec:IPC-Ach}). This protocol naturally serves also as an IPC-SI protocol. The GMPC protocol is based on the idea of non-uniform randomized partitioning, and generalizes our recently proposed protocol in~\cite{HKS2019Journal} for the PIR-CSI setting. Examples of the GMPC protocol are given in Section~\ref{sec:ex}.

\begin{remark}\label{rem:1}
\emph{The matching capacity of the IPC-SI and IPC-CSI settings shows that achieving individual privacy comes at no loss in capacity if the user has only \emph{one} random linear combination of $M$ random messages, instead of $M$ random messages separately as their side information.}	
\end{remark}

\begin{remark}\label{rem:2}
\emph{As shown in~\cite{HKRS2019}, for the IPIR-SI setting, the normalized download cost of $K-M\lfloor\frac{K}{M+D}\rfloor$ or $D\lceil \frac{K}{M+D}\rceil$ (depending on $K,M,D$) is achievable, where the \emph{normalized download cost} is defined as the download cost normalized by the entropy of a message. Comparing this with the result of Theorem~\ref{thm:IPC-SI}, one can see that, when individual privacy is required, one can privately compute a linear combination of multiple messages much more efficiently than retrieving them privately and linearly combining them locally.}	
\end{remark}

\begin{remark}\label{rem:3}
\emph{For the case of $M=0$, the capacity of both IPC-SI and IPC-CSI settings is equal to $\lceil\frac{K}{D}\rceil^{-1}$. Depending on the value of $D$, the capacity can be substantially larger than $\frac{1}{K}$, which was shown to be the capacity of single-server private computation where the privacy of both the demand support index set and the coefficients in the demand must be preserved~\cite{SJ2018}. For the case of $D=1$, both the IPC-SI and IPC-CSI problems reduce to the problems of PIR-SI~\cite{Kadhe2017} and PIR-CSI where the demanded message does not lie in the support of the side information~\cite{HKS2018}, respectively. The capacity of these settings were shown to be equal to $\lceil\frac{K}{M+1}\rceil^{-1}$, matching the results of Theorems~\ref{thm:IPC-SI} and~\ref{thm:IPC-CSI}.}	
\end{remark}

\subsection{JPC-SI and JPC-CSI}\label{subsec:JPC}
Theorem~\ref{thm:JPC-SI} lower bounds the capacity of JPC-SI for all $K,M,D$, and Theorem~\ref{thm:JPC-CSI} establishes a lower bound on the capacity of JPC-CSI for some values of $K,M,D$. 

\begin{theorem}[{\cite{HS2019PC}}]\label{thm:JPC-SI}
For the JPC-SI setting with $K$ messages, side information of size $M$, and demand support size $D$, the capacity is lower bounded by $(\lceil \frac{K-M-D}{\lfloor M/D\rfloor+1}\rceil+1)^{-1}$. 		
\end{theorem}

\begin{theorem}\label{thm:JPC-CSI}
For the JPC-CSI setting with $K$ messages, side information support size $M$, and demand support size $D$, the capacity is lower bounded by ${(\frac{K-M-D}{\lfloor M/D\rfloor+1}+1)^{-1}}$ if $\lfloor\frac{M}{D}\rfloor+1$ divides $K-M-D$. \end{theorem}

The capacity lower bound in Theorem~\ref{thm:JPC-SI} is achievable by a scalar-linear JPC-SI protocol, called \emph{Partition-and-Code with Interference Alignment (PC-IA)}, which we recently proposed in~\cite{HS2019PC}. The PC-IA protocol is applicable for all $K,M,D$, and relies on the idea of a probabilistic partitioning that allows the parts to overlap and have multiple blocks of interference that are aligned (for details, see~\cite{HS2019PC}). 

Theorem~\ref{thm:JPC-CSI}, which appears without proof, follows directly from a simple observation that the PC-IA protocol (with a slight modification in the choice of coefficients in the linear combinations that constitute the server's answer to the user's query) serves also as a scalar-linear JPC-CSI protocol for some values of $K,M,D$, particularly when the divisibility condition in the theorem's statement holds; however, the PC-IA protocol is not a JPC-CSI protocol in general. 
Examples of the PC-IA protocol for both cases are given in Section~\ref{sec:ex}. 

We have been able to design different scalar-linear JPC-CSI protocols for some other values of $K,M,D$; but the constructions are not universal and are limited to specific values of $K,M,D$, and hence not presented in this work. The extension of these constructions to arbitrary $K,M,D$ is a challenging open problem, and the focus of an ongoing work. 

\begin{remark}\label{rem:4}
\emph{As was shown in~\cite{HS2019PC}, when joint privacy is required, with the help of an uncoded side information the download cost for the private computation of one linear combination of multiple messages can be much lower than that of privately retrieving multiple messages and computing the linear combination of them. For instance, for $K$ even, when the user has $M=2$ messages as side information, for privately computing a linear combination of $D=2$ messages the normalized download cost is equal to $\frac{K}{2}-1$ (see Theorem~\ref{thm:JPC-SI}); whereas, privately retrieving $D=2$ messages incurs a normalized download cost of $\min\{{K-2},{K-\lfloor\frac{K}{3}\rfloor}\}$, which is significantly higher than $\frac{K}{2}-1$ (see~\cite[Theorem~2]{HKGRS:2018}). Surprisingly, the result of Theorem~\ref{thm:JPC-CSI} shows that for some values of $K,M,D$ (e.g., $K$ even and $M=D=2$), only \emph{one} linear combination of $M$ messages suffices to achieve the same normalized download cost (e.g., $\frac{K}{2}-1$). This is interesting because regardless of the values of $M$ and $D$,  when joint privacy is required, with the help of only one linear combination of $M$ messages the normalized download cost for retrieving $D$ messages is equal to $K-1$, which is much higher than, for instance, $\frac{K}{2}-1$.}	
\end{remark}

\begin{remark}\label{rem:5}
\emph{The capacity lower bounds in Theorems~\ref{thm:JPC-SI} and~\ref{thm:JPC-CSI} are tight for the cases of $D=1$ and $M=0$ (see~\cite{Kadhe2017,HKS2018}). We have been able to prove the tightness of these bounds for small values of $K,M,D$, particularly for $M=D=2$ and several values of $K$. Nevertheless, it remains open whether these lower bounds are tight for all $K,M,D$ in general.}
\end{remark}

\begin{remark}\label{rem:6}
\emph{The matching capacity lower bounds in Theorems~\ref{thm:JPC-SI} and~\ref{thm:JPC-CSI} raises an intriguing question whether, similar to the IPC-SI and IPC-CSI settings, the capacity of the JPC-SI and JPC-CSI settings are the same. We conjecture that the answer is affirmative for both linear and non-linear protocols.}
\end{remark}

\section{Proofs of Theorems~\ref{thm:IPC-SI} and~\ref{thm:IPC-CSI}}\label{sec:ProofsIPC}
Since any IPC-CSI protocol is an IPC-SI protocol, for the converse we only need to upper bound the rate of any IPC-SI protocol, whereas for the achievability it suffices to design an IPC-CSI protocol that achieves the rate upper bound. 

\subsection{Converse}\label{subsec:IPC-Conv}
\begin{lemma}\label{lem:IPC-Conv}
The rate of any IPC-SI protocol for $K$ messages, side information of size $M$, and demand support size $D$, is upper bounded by ${\lceil \frac{K}{M+D} \rceil}^{-1}$.	
\end{lemma}

\begin{proof}
To prove the lemma, we need to show that $H(\mathbf{A})\geq \lceil\frac{K}{M+D}\rceil L$. Take arbitrary $W\in \mathcal{K}_D$, $V\in \mathcal{C}_D$, $S\in \mathcal{K}_M$ such that $S\cap W = \emptyset$. By a simple application of the chain rule of entropy, one can show that
\begin{equation}\label{eq:CSIlineI20}
H(\mathbf{A})\geq H(\mathbf{Z})+H(\mathbf{A}|\mathbf{Q},\mathbf{X}_S,\mathbf{Z}),	
\end{equation} where $\mathbf{Z}\triangleq \mathbf{Z}^{[W,V]}$. Note that $H(\mathbf{Z})=L$. We consider two cases: (i) $W \cup S = \mathcal{K}$, and (ii) $W \cup S\neq \mathcal{K}$. In the case (i), we have $M = K-D$, and $\lceil \frac{K}{M+D}\rceil L = L$; and hence,~\eqref{eq:CSIlineI20} implies that $H(\mathbf{A})\geq H(\mathbf{Z}) = L$, as was to be shown. 

In the case (ii), we further lower bound $H(\mathbf{A}|\mathbf{Q},\mathbf{X}_S,\mathbf{Z})$ as follows. Choose an arbitrary message, say $\mathbf{X}_{i_1}$, for some $i_1\not\in W\cup S$. By the result of Lemma~\ref{lem:1}, there exist ${W_1\in \mathcal{K}_D}$, $i_1\in W_1$, $V_1\in \mathcal{C}_D$, and $S_1\in \mathcal{K}_M$, ${S_1\cap W_1 = \emptyset}$, such that $H(\mathbf{Z}_1 |\mathbf{A},\mathbf{Q},\mathbf{X}_{S_1}) = 0$, or in turn, $H(\mathbf{Z}_1 |\mathbf{A},\mathbf{Q},\mathbf{X}_S,\mathbf{Z},\mathbf{Z}_1) = 0$, where $\mathbf{Z}_1\triangleq \mathbf{Z}^{[W_1,V_1]}$. Thus, 
\begin{align}
H(\mathbf{A}|\mathbf{Q},\mathbf{X}_S,\mathbf{Z}) &\geq 
H(\mathbf{A}|\mathbf{Q},\mathbf{X}_S,\mathbf{Z},\mathbf{X}_{S_1}) \nonumber\\ 
& \quad +H(\mathbf{Z}_1|\mathbf{A},\mathbf{Q},\mathbf{X}_S,\mathbf{Z},\mathbf{X}_{S_1})\nonumber \\
&=H(\mathbf{Z}_{1}|\mathbf{Q},\mathbf{X}_S,\mathbf{Z},\mathbf{X}_{S_1})\nonumber \\ 
& \quad +H(\mathbf{A}|\mathbf{Q},\mathbf{X}_S,\mathbf{Z},\mathbf{X}_{S_1},\mathbf{Z}_{1})\nonumber\\ 
& = H(\mathbf{Z}_{1})\nonumber \\ & \quad +H(\mathbf{A}|\mathbf{Q},\mathbf{X}_S,\mathbf{Z},\mathbf{X}_{S_1},\mathbf{Z}_{1})\label{eq:CSIlineI21}
\end{align} where $\mathbf{Z}_{1}$ and $(\mathbf{Q},\mathbf{X}_S,\mathbf{Z},\mathbf{X}_{S_1})$ are independent because $i_1\in W_1$ and ${i_1\not\in {W}\cup {S}\cup {S}_1}$. 
Let $n \triangleq \lceil\frac{K}{M+D}\rceil$. Using Lemma~\ref{lem:1} recursively, it can be shown that for all ${1\leq k<n}$ there exist $i_1, \dots,i_k \in \mathcal{K}$, ${W}_1,\dots,{W}_{k}\in \mathcal{K}_D$, $V_1,\dots,V_k\in \mathcal{C}_D$, and ${S_1,\dots,S_{k}\in \mathcal{K}_M}$ satisfying ${i_l \in W_l}$, ${S}_l\cap {W}_l = \emptyset$, ${i_l\not\in \cup_{j=1}^{l-1} ({W}_{j}\cup {S}_{j}) \cup ({W}\cup {S})}$ for all $1\leq l\leq k$, such that \[H(\mathbf{Z}_{k}|\mathbf{A},\mathbf{Q},\mathbf{X}_S,\mathbf{Z},\mathbf{X}_{S_1},\mathbf{Z}_{1},\dots,\mathbf{X}_{S_{k-1}},\mathbf{Z}_{k-1},\mathbf{X}_{S_k})=0,\] where $\mathbf{Z}_l \triangleq \mathbf{Z}^{[W_l,V_l]}$ for all $1\leq l\leq k$. Obviously, $\left|\cup_{j=1}^{k-1} ({W}_j \cup {S}_j) \cup ({W}\cup {S})\right|\leq (M+D)k$ for all ${1\leq k< n}$. Applying the same technique as in~\eqref{eq:CSIlineI21}, it can then be shown that for all $1\leq k< n$, we have
\begin{align*}
& H(\mathbf{A}|\mathbf{Q},\mathbf{X}_S,\mathbf{Z},\mathbf{X}_{S_1},\mathbf{Z}_{1},\dots,\mathbf{X}_{S_{k-1}},\mathbf{Z}_{k-1})\nonumber \\ 
& \quad \geq H(\mathbf{Z}_k) +H(\mathbf{A}|\mathbf{Q},\mathbf{X}_S,\mathbf{Z},\mathbf{X}_{S_1},\mathbf{Z}_{1},\dots,\mathbf{X}_{S_k},\mathbf{Z}_{k}).
\end{align*} Putting together these lower bounds for all $k$, we have
\begin{align}
H(\mathbf{A}|\mathbf{Q},\mathbf{X}_S,\mathbf{Z}) &\geq \sum_{k=1}^{n-1} H(\mathbf{Z}_k) = (n-1)L, \label{eq:CSIlineI22}
\end{align} since $\mathbf{Z}_1,\dots,\mathbf{Z}_{n-1}$ are independent by the choice of $i_1,\dots,i_{n-1}$ in the construction. 
Combining~\eqref{eq:CSIlineI20} and~\eqref{eq:CSIlineI22}, we have $H(\mathbf{A})\geq nL = \lceil\frac{K}{M+D}\rceil L$, as was to be shown.	
\end{proof}

\subsection{Achievability}\label{subsec:IPC-Ach}
For the ease of notation, we define $n\triangleq \lceil\frac{K}{M+D}\rceil$, ${m\triangleq n(M+D)-K}$, and ${r \triangleq M+D-m}$. 

\vspace{0.25cm}
\textbf{Generalized Modified Partition-and-Code (GMPC) Protocol:} This protocol consists of three steps as follows: 

\emph{Step~1:} Let $I_l \triangleq \{{(l-1)(M+D)+1},\dots,l(M+D)\}$ for $1\leq l< n$, and let $I_n\triangleq \{1,\dots,m,(n-1)(M+D)+1,\dots,K\}$. Note that $I_1\cap I_n = \{1,\dots,m\}$. 

First, the user constructs a random permutation $\pi$ on $\mathcal{K}$ as follows. With probability $\alpha\triangleq \frac{m+2r}{K}$, the user chooses $l^{*}\in\{1,n\}$ uniformly at random; otherwise, with probability $1-\alpha$, the user randomly chooses $l^{*}\in\{2,\dots,n-1\}$. 

If $l^{*}\in \{1,n\}$, with probability $\beta$ (or $1-\beta$) where the choice of $\beta$ will be specified shortly, the user assigns $\mu\triangleq\min\{D,m\}$ (or $D-\rho\triangleq D-\min\{D,r\}$)  randomly chosen indices from $W$ and $m-\mu$ (or $m-D+\rho$) randomly chosen indices from $S$ to $\{\pi(j): 1\leq j\leq m\}$ at random, and randomly assigns the rest of the indices in $W\cup S$ to $\{\pi(j): j\in I_{l^{*}}\setminus \{1,\dots,m\}\}$. Otherwise, if $l^{*}\in \{2,\dots,n-1\}$, the user randomly assigns the $M+D$ indices in $W\cup S$ to $\{\pi(j): j\in I_{l^{*}}\}$. Then, the user assigns the (not-yet-assigned) indices in $\mathcal{K}\setminus (W\cup S)$ to $\{\pi(j): j\not\in I_{l^{*}}\}$. 

The value of $\beta$, which is carefully chosen in order to satisfy the individual privacy condition, depends on the values of $D, m, r$: 
\[\beta \triangleq 
\begin{cases}
\frac{m}{m+2r}, &  D\leq m, D\leq r,\\
\frac{D}{m+2r}, &  D> m, D\leq r,\\
1-\frac{2D}{m+2r}, &  D\leq m, D> r,\\
\frac{r}{M}\left(1- \frac{2D}{m+2r}\right), &  D> m, D> r.\\
\end{cases}
\]

Next, the user constructs $n$ ordered sets $Q'_1,\dots,Q'_n$, each of size $M+D$, defined as $Q'_k \triangleq \{\pi(j): j\in I_l\}$; and constructs an ordered multiset $Q''$ of size $M+D$, defined as $Q'' \triangleq \{c_{j}: j\in I_{l^{*}}\}$ where $c_{j}= v_{\pi(j)}$ or $c_{j}= u_{\pi(j)}$ when $\pi(j)\in W$ or $\pi(j)\in S$, respectively. Recall that $v_{\pi(j)}$ or $u_{\pi(j)}$ is the coefficient of the message $X_{\pi(j)}$ in the user's demand or side information, respectively. 

The user then constructs $Q_l = (Q'_l,Q'')$ for ${1\leq l\leq n}$, and sends the query $Q^{[W,V,S,U]} = \{Q_{1},\dots,Q_{n}\}$ to the server. 

\emph{Step~2:} By using $Q_l=(Q'_l,Q'')$'s, the server computes $A_l$'s, defined as $A_{l} \triangleq \sum_{j=1}^{M+D} c_{i_j} X_{i_j}$ where $Q'_{l} = \{i_1,\dots,i_{M+D}\}$ and $Q'' = \{c_{i_1},\dots,c_{i_{M+D}}\}$, and sends the answer $A^{[W,V,S,U]}=\{A_{1},\dots,A_{n}\}$ back to the user.

\emph{Step~3:} Upon receiving the server's answer, the user retrieves the demand $Z^{[W,V]}$ by subtracting off the contribution of the side information $Y^{[S,U]}$ from $A_{l^{*}}=Z^{[W,V]}+Y^{[S,U]}$.

\begin{lemma}\label{lem:IPC-Ach}
The GMPC protocol is a scalar-linear IPC-CSI protocol, and achieves the rate ${\lceil \frac{K}{M+D} \rceil}^{-1}$.	
\end{lemma}

\begin{proof}
The rate and the scalar-linearity of the GMPC protocol are obvious from the construction. Clearly, the recoverability condition is also satisfied. 

To prove that the GMPC protocol satisfies the individual privacy condition, we need to show that for any given query $Q$ generated by the protocol, for all $i\in \mathcal{K}$, it holds that \[\Pr(i\in \mathbf{W}|\mathbf{Q}=Q) = \Pr(i\in \mathbf{W})=\frac{D}{K},\] noting that $\mathbf{W}$ is distributed uniformly over $\mathcal{K}_D$.   

Fix an arbitrary $i\in \mathcal{K}$. We consider the following three different cases separately: (i) $\pi^{-1}(i)\in \{1,\dots,m\}$; (ii) $\pi^{-1}(i)\in I_l\setminus \{1,\dots,m\}$ for some $l\in \{1,n\}$; and (iii) $\pi^{-1}(i)\in I_l$ for some $l\not\in \{1,n\}$, where $\pi^{-1}(i) = j$ if and only if $\pi(j)=i$.  

First, consider the case (i). In this case, we have
\begin{align*} 
& \Pr(i\in \mathbf{W}|\mathbf{Q}=Q) \\ 
& \quad = \sum_{l\in \{1,n\}} \Pr(i\in \mathbf{W}, l^{*} = l | \mathbf{Q}=Q)\\
& \quad = \sum_{l\in \{1,n\}} \Pr(l^{*} = l | \mathbf{Q}=Q)\times\Pr(i\in \mathbf{W} | \mathbf{Q}=Q, l^{*}=l)\\
& \quad = 2\left(\frac{1}{2}\times\alpha\left(\beta\times \frac{\binom{m-1}{\mu-1}}{\binom{m}{\mu}}+ (1-\beta)\times\frac{\binom{m-1}{D-\rho-1}}{\binom{m}{D-\rho}} \right)\right)\\
&\quad = 
\begin{cases}
\alpha\beta\left(\frac{D}{m}\right), & D\leq m, D\leq r,\\ 	
\alpha\beta, & D>m, D\leq r,\\
\alpha\left(\beta\left(\frac{D}{m}\right)+(1-\beta)\left(\frac{D-r}{m}\right)\right), & D\leq m, D>r,\\
\alpha\left(\beta + (1-\beta)\left(\frac{D-r}{m}\right))\right), & D>m, D>r,
\end{cases}\\
& \quad = \frac{D}{K},
\end{align*} for our choice of $\beta$ for each range of values of $D,m,r$.   

Next, consider the case (ii). In this case, we have
\begin{align*} 
& \Pr(i\in \mathbf{W}|\mathbf{Q}=Q) \\ 
& \quad = \Pr(i\in \mathbf{W}, l^{*} = l | \mathbf{Q}=Q)\\
& \quad = \Pr(l^{*} = l | \mathbf{Q}=Q)\times \Pr(i\in \mathbf{W} | \mathbf{Q}=Q, l^{*}=l)\\
& \quad = \frac{1}{2}\times\alpha\left(\beta\times \frac{\binom{r-1}{D-\mu-1}}{\binom{r}{D-\mu}}+ (1-\beta)\times\frac{\binom{r-1}{\rho-1}}{\binom{r}{\rho}} \right)\\
&\quad = 
\begin{cases}
\frac{\alpha}{2}(1-\beta)\left(\frac{D}{r}\right), & D\leq m, D\leq r,\\ 	
\frac{\alpha}{2}\left(\beta\left(\frac{D-m}{r}\right)+(1-\beta)\left(\frac{D}{r}\right)\right), & D>m, D\leq r,\\
\frac{\alpha}{2}(1-\beta), & D\leq m, D>r,\\
\frac{\alpha}{2}\left(\beta\left(\frac{D-m}{r}\right) + (1-\beta)\right), & D>m, D>r,
\end{cases}\\
& \quad = \frac{D}{K}, 
\end{align*} for the choices of $\beta$ specified earlier. 

Lastly, consider the case (iii). In this case, we have 
\begin{align*} 
& \Pr(i\in \mathbf{W}|\mathbf{Q}=Q) \\ 
& \quad = \Pr(i\in \mathbf{W}, l^{*} = l | \mathbf{Q}=Q)\\
& \quad = \Pr(l^{*} = l | \mathbf{Q}=Q)\Pr(i\in \mathbf{W} | \mathbf{Q}=Q, l^{*}=l)\\
& \quad = \frac{1}{n-2}\times (1-\alpha)\left(\frac{D}{M+D}\right)\\
& \quad = \left(\frac{M+D}{K-m-2r}\right)\left(\frac{K-m-2r}{K}\right)\left(\frac{D}{M+D}\right)\\
& \quad = \frac{D}{K}.  
\end{align*} This completes the proof. 
\end{proof}

\section{Examples}\label{sec:ex}

\subsection{GMPC Protocol}
This section illustrates two examples of the GMPC protocol for $M=D=2$ and $K\in \{11,12\}$. 

\begin{example}\label{ex:IPC-1}
\emph{Consider a scenario where the server has ${K=12}$ messages $X_1,\dots,X_{12}\in \mathbb{F}_{7}$, and the user demands the linear combination $Z=X_1+3X_2$ with support size $D=2$ and has a coded side information $Y = 5X_3+X_4$ with support size $M=2$. For this example, $W = \{1,2\}$, $V = \{v_1,v_2\} = \{1,3\}$, $S= \{3,4\}$, and ${U =\{u_3,u_4\}= \{5,1\}}$.}

\emph{For this example, the protocol's parameters are as follows: $n =3$, $m=0$, $r=4$, $\alpha = \frac{2}{3}$, $\beta = \frac{1}{4}$, $\mu = 0$ and $\rho=2$.} 

\emph{Let $I_1 = \{1,2,3,4\}$, $I_2 = \{5,6,7,8\}$, and $I_3 = \{9,10,11,12\}$. First, the user constructs a permutation $\pi$ of $\{1,\dots,12\}$ as follows. With probability $\alpha = \frac{2}{3}$, the user randomly chooses $l^{*}\in \{1,3\}$, or with probability $1-\alpha = \frac{1}{3}$, the user chooses $l^{*}=2$. Note that for this example, $l^{*}$ is equally likely to be any of the indices in $\{1,2,3\}$. Suppose that the user chooses $l^{*}=1$. Since, for this example, $\mu=0$, $D-\rho=0$, $m-\mu=0$, and $m-D+\rho=0$, the user randomly assigns all indices in $W\cup S = \{1,2,3,4\}$ to $\{\pi(j): j\in I_{1}\}$; say, $\pi(1) = 2$, $\pi(2)=4$, $\pi(3)=1$, and $\pi(4)=3$. Then, the user randomly assigns the (not-yet-assigned) indices in $\{5,\dots,12\}$ to $\{\pi(j): j\not\in I_1\}$; say, $\pi(5) = 10$, $\pi(6)=8$, $\pi(7)=6$, $\pi(8)=5$, $\pi(9)=11$, $\pi(10)=9$, $\pi(11)=12$, and $\pi(12)=7$. Thus, the permutation $\pi$ maps $\{1,\dots,12\}$ to $\{2,4,1,3,10,8,6,5,11,9,12,7\}$.} 

\emph{Next, the user constructs the ordered sets $Q'_1 =\{\pi(j): j\in I_1\} = \{2,4,1,3\}$, $Q'_2 =\{\pi(j): j\in I_2\} = \{10,8,6,5\}$, and $Q'_3 =\{\pi(j): j\in I_3\} = \{11,9,12,7\}$; and constructs the ordered multiset $Q'' =\{c_{j}: j\in I_{1}\} = \{c_1,c_2,c_3,c_4\} = \{v_2,u_4,v_1,u_3\} = \{3,1,1,5\}$. The user then constructs $Q_1 = (Q'_1,Q'') = (\{2,4,1,3\},\{3,1,1,5\})$, $Q_2 = (Q'_2,Q'') = (\{10,8,6,5\},\{3,1,1,5\})$, and $Q_3 = (Q'_3,Q'') = (\{11,9,12,7\},\{3,1,1,5\})$; and sends the query $Q = \{Q_1,Q_2,Q_3\}$ to the server.} 

\emph{The server computes $A_1 = 3X_2+X_4+X_1+5X_3$, $A_2 = 3X_{10}+X_8+X_6+5X_5$, and $A_3 = 3X_{11}+X_9+X_{12}+5X_7$; and sends the answer $A = \{A_1,A_2,A_3\}$ back to the user. The user then subtracts $Y = 5X_3+X_4$ from $A_{l^{*}}=A_1 = 3X_2+X_4+X_1+5X_3$, and recovers $Z = X_1+3X_2$.} 

\emph{To prove that the individual privacy condition is satisfied in this example, we need to show that the probability of every message $X_i$ to be one of the two messages in $X_W$ is $\frac{2}{12}=\frac{1}{6}$. From the perspective of the server, $l^{*}$ is $1$, $2$, or $3$, each with probability $\frac{1}{3}$. Given $l^{*}=l$, every one of the $6$ pairs of messages in the support of $A_l$ are the two messages in $X_W$ with probability $\frac{1}{6}$. Since every message in the support of $A_l$ belongs to $3$ pairs of messages, the probability of any message in the support of $A_l$ to be one of the two messages in $X_W$ is $3\times\frac{1}{6} = \frac{1}{2}$. Thus, the probability of any message $X_i$ to belong to $X_W$ is $\frac{1}{3}\times \frac{1}{2} = \frac{1}{6}$.} 
\end{example}

\begin{example}\label{ex:IPC-2}
\emph{Consider the scenario in Example~\ref{ex:IPC-1} (i.e., $W=\{1,2\}$, $V=\{1,3\}$, $S=\{3,4\}$, and $U=\{5,1\}$), except when the server has ${K=11}$ messages $X_1,\dots,X_{11}\in \mathbb{F}_{7}$.}

\emph{The protocol's parameters for this example are as follows: $n =3$, $m=1$, $r=3$, $\alpha = \frac{7}{11}$, $\beta = \frac{2}{7}$, $\mu = 1$ and $\rho=2$.} 

\emph{Let $I_1 = \{1,2,3,4\}$, $I_2 = \{5,6,7,8\}$, and $I_3 = \{1,9,10,11\}$. The user first constructs a permutation $\pi$ of $\{1,\dots,11\}$ as follows. With probability $\alpha = \frac{7}{11}$, the user randomly chooses $l^{*}\in \{1,3\}$, or with probability $1-\alpha = \frac{4}{11}$, the user chooses $l^{*}=2$. For this example, $l^{*}$ is equal to $1$, $2$, or $3$, with probability $\frac{7}{22}$, $\frac{4}{11}$, or $\frac{7}{22}$, respectively. Suppose that the user chooses $l^{*}=1$. With probability $\beta = \frac{2}{7}$ (or $1-\beta = \frac{5}{7}$), the user assigns $\mu=1$ randomly chosen index from $W = \{1,2\}$ (or $m-D+\rho = 1$ randomly chosen index from $S = \{3,4\}$), say the index $2$, to $\pi(1)$, i.e., $\pi(1)=2$; and randomly assigns the rest of the indices in $W\cup S = \{1,2,3,4\}$, i.e., $\{1,3,4\}$, to $\{\pi(j): j\in I_{1}\setminus \{1\}\}$; say $\pi(2)=4$, $\pi(3)=1$, and $\pi(4)=3$. Then, the user randomly assigns the (not-yet-assigned) indices in $\{5,\dots,11\}$ to $\{\pi(j): j\not\in I_1\}$; say, $\pi(5) = 10$, $\pi(6)=8$, $\pi(7)=6$, $\pi(8)=5$, $\pi(9)=11$, $\pi(10)=9$, and $\pi(11)=7$. Thus, the permutation $\pi$ maps $\{1,\dots,11\}$ to $\{2,4,1,3,10,8,6,5,11,9,7\}$.} 

\emph{Next, the user constructs the ordered sets $Q'_1 =\{\pi(j): j\in I_1\} = \{2,4,1,3\}$, $Q'_2 =\{\pi(j): j\in I_2\} = \{10,8,6,5\}$, and $Q'_3 =\{\pi(j): j\in I_3\} = \{2,11,9,7\}$; and constructs the ordered multiset $Q'' =\{c_{j}: j\in I_{1}\} = \{c_1,c_2,c_3,c_4\} = \{v_2,u_4,v_1,u_3\} = \{3,1,1,5\}$. The user then constructs $Q_1 = (Q'_1,Q'') = (\{2,4,1,3\},\{3,1,1,5\})$, $Q_2 = (Q'_2,Q'') = (\{10,8,6,5\},\{3,1,1,5\})$, and $Q_3 = (Q'_3,Q'') = (\{2,11,9,7\},\{3,1,1,5\})$; and sends the query $Q = \{Q_1,Q_2,Q_3\}$ to the server.} 

\emph{The server computes $A_1 = 3X_2+X_4+X_1+5X_3$, $A_2 = 3X_{10}+X_8+X_6+5X_5$, and $A_3 = 3X_{2}+X_{11}+X_{9}+5X_7$; and sends the answer $A = \{A_1,A_2,A_3\}$ back to the user. The user then subtracts $Y = 5X_3+X_4$ from $A_{l^{*}}=A_1 = 3X_2+X_4+X_1+5X_3$, and recovers $Z = X_1+3X_2$.} 

\emph{Now, we show that the individual privacy condition is satisfied for this example. We need to verify that every message $X_i$ belongs to $X_W$ with probability $\frac{2}{11}$. From the server's perspective, $l^{*}$ is $1$, $2$, or $3$ with probability $\frac{7}{22}$, $\frac{4}{11}$, or $\frac{7}{22}$, respectively. First, consider the message $X_2$. Given $l^{*} = 1$ (or $l^{*} = 3$), the message $X_2$, which belongs to the support of both $A_1$ and $A_3$, is one of the two messages in $X_W$ with probability $\frac{2}{7}$; whereas for $l^{*} = 2$, the message $X_2$ cannot belong to $X_W$. Thus, the probability of the message $X_2$ to belong to $X_W$ is $2\times \frac{7}{22}\times \frac{2}{7}=\frac{2}{11}$. Now, consider the message $X_1$. Given $l^{*} = 1$, the message $X_1$ belongs to $X_W$ with probability $\frac{2}{7}\times \frac{1}{3}+\frac{5}{7}\times\frac{2}{3} = \frac{4}{7}$. This is because for $X_1$ being one of the two messages in $X_W$ given $l^{*}=1$, either (i) $X_2$ belongs to $X_W$, which has probability $\frac{2}{7}$, and $X_1$ is the other message in $X_W$, which has probability $\frac{1}{3}$ (given $X_2$ belonging to $X_W$), or (ii) $X_2$ does not belong to $X_W$, which has probability $\frac{5}{7}$, and one of the pairs $X_1,X_3$ or $X_1,X_4$ are the two messages in $X_W$, which has probability $\frac{2}{3}$ (given $X_2$ not belonging to $X_W$). Given $l^{*}=2$ or $l^{*}=3$, the message $X_1$ cannot be one of the two messages in $X_W$. Thus, the probability of the message $X_1$ to belong to $X_W$ is $\frac{7}{22}\times \frac{4}{7} = \frac{2}{11}$. Similarly, one can show that any message $X_i$ belongs to $X_W$ with probability $\frac{2}{11}$.} 
\end{example}

\subsection{PC-IA Protocol}
In this section, we give two examples for the PC-IA protocol for $M=D=2$ and $K\in \{11,12\}$. Example~\ref{ex:JPC-1} shows an instance where the PC-IA is a JPC-CSI protocol, whereas Example~\ref{ex:JPC-2} shows an instance for which the PC-IA fails as a JPC-CSI protocol. 

\begin{example}\label{ex:JPC-1}
\emph{Consider the scenario in Example~\ref{ex:IPC-1}, except when joint privacy is required, instead of individual privacy.} 

\emph{The protocol's parameters for this example are as follows (for details, see~\cite{HS2019PC}): ${s = 2}$, $n=5$, $m = 6$, $r = 0$, and $t = 1$, and $\{x_1,x_2,x_3,x_4,x_5,y_0,y_1\}=\{0,1,\dots,6\}$.} 

\emph{First, the user creates $m=6$ ordered sets $B_1,\dots,B_6$, where $B_j=\{-,-\}$ for all $j$, i.e., $B_j$ has two slots to be filled with elements from $\{1,\dots,12\}$. The user then randomly places the $D=2$ indices in $W=\{1,2\}$ into two slots; say, $B_2=\{2,-\}$, $B_3=\{-,1\}$, and $B_1, B_4, B_5, B_6$ remain empty. Since $B_2$ and $B_3$ contain some indices from $W$, the user fills $B_2$ and $B_3$, each with a randomly chosen index from $S=\{3,4\}$; say, $B_2=\{2,4\}$, and $B_3=\{3,1\}$. Next, the user randomly places the remaining indices $5,\dots,12$ into the remaining slots, and fills $B_1, B_4,B_5,B_6$; say $B_1=\{11,6\}$, $B_4=\{10,12\}$, $B_5 = \{8,9\}$, and $B_6=\{7,5\}$.} 

\emph{The user then constructs $n=5$ ordered sets $Q_1,\dots,Q_5$, where $Q_i = \{B_1,B_{1+i}\}$. That is, $Q_1 = \{11,6,2,4\}$, $Q_2=\{11,6,3,1\}$, $Q_3=\{11,6,10,12\}$, $Q_4=\{11,6,8,9\}$, and $Q_5 = \{11,6,7,5\}$. Next, the user creates $n=5$ ordered multisets $Q'_1,\dots,Q'_5$, defined as $Q'_i=\{C_{i,1},C_{i,1+i}\}$, where $C_{i,1} = \{\alpha_{1,1}\omega_{i,1},\alpha_{1,2}\omega_{i,1}\}$ and $C_{i,1+i}=\{\alpha_{1+i,1},\alpha_{1+i,2}\}$, where $\omega_{i,j}= {1}/{(x_i-y_j)}$; and the values of $\alpha_{j,k}$'s are specified shortly. For this example, the user constructs $Q'_1=\{2,5,3,1\}$, $Q'_2=\{1,6,1,3\}$, $Q'_3=\{3,1,3,3\}$, $Q'_4=\{4,3,4,1\}$, and $Q'_5 = \{6,1,3,2\}$.} 

\emph{The procedure for choosing $\alpha_{j,k}$'s is described below. First, the user finds: (i) the set $J$ of indices $j$ such that $B_j$ contains some indices from $W$; (ii) the minimal set $I$ (with highest lexicographical order) of indices of $Q_i$'s such that $\cup_{i\in I} Q_i$ contains all indices in $W$; and (iii) the set $H$ of $|I|-1 = 1$ largest indices in $\{1,\dots,t\}\setminus J$; for this example, $J=\{2,3\}$, $I=\{1,2\}$, and $H=\{1\}$. Then, the user forms the matrix $T=(\omega_{i,j})_{i\in I,j\in H}=[\omega_{1,1},\omega_{2,1}]^{\mathsf{T}} = [1,4]^{\mathsf{T}}$, and chooses $c_1=1$ and $c_2=-{\omega_{1,1}}/{\omega_{2,1}}=5$ such that $[c_1,c_2]\cdot T=0$. The user then selects $\alpha_{j,k}$'s as follows: for $j\in J = \{2,3\}$ and $k\in \{1,\dots,s\}=\{1,2\}$ such that the $k$th element of $B_j$, say, $l$, belongs to $W$ (or $S$), the user selects $\alpha_{j,k}=v_{l}/\sum_{i\in I} c_i\omega_{i,j}$ (or $\alpha_{j,k}=u_{l}/\sum_{i\in I} c_i\omega_{i,j}$) if $1\leq j\leq t=1$, and selects $\alpha_{1+i,k}=v_{l}/c_i$ (or $\alpha_{1+i,k}=u_{l}/c_i$), where $v_{l}$ (or $u_l$) is the coefficient of the message $X_{l}$ in the user's demand $Z$ (or side information $Y$). For this example, $\alpha_{2,1} = v_2/c_1 = 3$ and $\alpha_{2,2}=u_4/c_1 = 1$ (the first element in $B_2$ and the second element in $B_3$ are the demand support indices $2$ and $1$, respectively); and $\alpha_{3,1} = u_3/c_2 = 1$ and $\alpha_{3,2} = v_1/c_2 = 3$ (the second element in $B_2$ and the first element in $B_3$ are the side information support indices $4$ and $3$, respectively). The user selects the rest of $\alpha_{j,k}$'s from $\mathbb{F}_q\setminus \{0\} = \{1,\dots,6\}$ at random; say, $\alpha_{1,1}=2$, $\alpha_{1,2}=5$, $\alpha_{4,1}=3$, $\alpha_{4,2}=3$, $\alpha_{5,1}=4$, $\alpha_{5,2}=1$, $\alpha_{6,1}=3$, and $\alpha_{6,2}=2$. These choices of $\alpha_{j,k}$'s yield the ordered multisets $Q'_1,\dots,Q'_5$ defined earlier.}

\emph{The user then sends to the server
 \begin{align*}
(Q_1,Q'_1) &=(\{11,6,2,4\},\{2,5,3,1\}),\\
(Q_2,Q'_2) &=(\{11,6,3,1\},\{1,6,1,3\}),\\
(Q_3,Q'_3) &=(\{11,6,10,12\},\{3,1,3,3\}),\\
(Q_4,Q'_4) &=(\{11,6,8,9\},\{4,3,4,1\}),\\
(Q_5,Q'_5) &=(\{11,6,7,5\},\{6,1,3,2\}),
\end{align*} and the server sends the user back 
\begin{align*}
A_1 &=2X_{11}+5X_{6}+3X_2+X_4,\\
A_2 &=X_{11}+6X_{6}+X_3+3X_1,\\
A_3 &=3X_{11}+X_{6}+3X_{10}+3X_{12},\\
A_4 &= 4X_{11}+3X_{6}+4X_{8}+X_9,\\
A_5 &= 6X_{11}+X_{6}+3X_{7}+2X_5.
\end{align*} Then, the user computes $c_1A_1+c_2A_2=A_1+ 5A_2= X_1+3X_2+5X_3+X_4$; and subtracting off the contribution of $Y = 5X_3+X_4$, recovers $Z = X_1+3X_2$.} 

\emph{To show that the joint privacy condition is satisfied in this example, we need to prove that any pair of messages is equally likely to be in $X_W$. As an example, consider the pair of messages $X_1$ and $X_2$. According to the supports of $A_1$ and $A_2$, $X_1$ and $X_2$ belong to $X_W$ if and only if $X_3$ and $X_4$ are the two messages in $X_S$. This is because $X_1$ belongs only to the support of $A_2$, and $X_2$ belongs only to the support of $A_1$; and by the protocol, one of the messages in $X_S$ (in this case, $X_3$) must be paired with $X_1$, and the other message in $X_S$ (in this case, $X_4$) must be paired with $X_2$. Note that $X_6$ and $X_{11}$ are aligned in $A_1$ and $A_2$; and they can be canceled by linearly combining $A_1$ and $A_2$. Moreover, there exists a unique such linear combination of $A_1$ and $A_2$, i.e., $c_1A_1+c_2A_2$, where the coefficient of $A_1$, i.e., $c_1$, is equal to $1$. Note that by the protocol, the coefficient of the least-indexed $A_i$ in the linear combination of $A_i$'s being used by the user in the recovery process (in this case, the coefficient $c_1$ of $A_1$ in the linear combination $c_1A_1+c_2A_2$) is always chosen to be equal to $1$.} 

\emph{Next, consider a different pair of messages, say, $X_2$ and $X_6$. According to the support of $A_1$, the messages $X_2$ and $X_6$ belong to $X_W$ if and only if the messages $X_{11}$ and $X_4$ belong to $X_S$. In this case, $A_1$ is the unique linear combination of $A_i$'s (with the least-indexed $A_i$ having coefficient $1$) whose support contains $X_2,X_6,X_{11},X_{4}$. As another example, consider the pair of messages $X_6$ and $X_{11}$. By the protocol, $X_{11}$ and $X_6$ belong to $X_W$ if and only if $X_2$ and $X_4$ belong to $X_S$. This is because $A_1$ is the least-indexed $A_i$ whose support contains $X_{11}$ and $X_6$.} 

\emph{Similarly as above, one can verify that from the perspective of the server, there is a unique way to recover a linear combination of any two messages. This observation, together with the fact that by the protocol the two messages in $X_W$ were placed randomly in the support of $A_i$'s, show that any pair of messages is equally likely to be in $X_W$.}
\end{example}

\begin{example}\label{ex:JPC-2}
\emph{Consider the scenario in Example~\ref{ex:IPC-2}, except when, instead of individual privacy, joint privacy is required.}

\emph{Following the procedure in the PC-IA protocol, without specifying the choice of nonzero coefficients (denoted by `$*$'), the server's answer to the user's query will have one of the following structures (up to a permutation of $1$ and $2$; a permutation of $3$ and $4$; and a permutation of $5,\dots,11$):}

\begin{enumerate}
\item[] Case (i):
\begin{align*}
A_1 &=*X_{5}+*X_{1}+*X_3,\\
A_2 &=*X_{5}+*X_{2}+*X_{4},\\
A_3 &=*X_{5}+*X_{6}+*X_{7},\\
A_4 &= *X_{5}+*X_{8}+*X_{9},\\
A_5 &= *X_{5}+*X_{10}+*X_{11}.
\end{align*}
\item[] Case (ii): 
\begin{align*}
A_1 &=*X_{3}+*X_{1}+*X_2,\\
A_2 &=*X_{3}+*X_{4}+*X_{5},\\
A_3 &=*X_{3}+*X_{6}+*X_{7},\\
A_4 &= *X_{3}+*X_{8}+*X_{9},\\
A_5 &= *X_{3}+*X_{10}+*X_{11}.
\end{align*}
\item[] Case (iii): 
\begin{align*}
A_1 &=*X_{1}+*X_{2}+*X_3,\\
A_2 &=*X_{1}+*X_{4}+*X_{5},\\
A_3 &=*X_{1}+*X_{6}+*X_{7},\\
A_4 &= *X_{1}+*X_{8}+*X_{9},\\
A_5 &= *X_{1}+*X_{10}+*X_{11}.
\end{align*} 	
\end{enumerate}

\emph{In either of these cases, one can easily verify that for any $D=2$ messages, say, $X_{i_1}$ and $X_{i_2}$, there exists a linear combination of $A_i$'s whose support includes $X_{i_1}$ and $X_{i_2}$, and has size at least $D=2$ and at most $M+D=4$. (Recall that this property is required for any scalar-linear JPC-SI protocol.) For instance, consider the two messages $X_1$ and $X_2$. In the case (i), by linearly combining $A_1$ and $A_2$ in such a way that $X_5$ is canceled, one can recover a linear combination of $X_1$, $X_2$, $X_3$, and $X_4$. In both cases (ii) and (iii), the support of $A_1$ contains $X_1$ and $X_2$, and has size $3$ ($<4$). On the other hand, for a scalar-linear JPC-CSI protocol, a stronger requirement needs to be satisfied: for any two messages $X_{i_1}$ and $X_{i_2}$, there must exist a linear combination of $A_i$'s whose support includes $X_{i_1}$ and $X_{i_2}$, and has size exactly equal to $M+D=4$. For the two messages $X_1$ and $X_2$, however, neither case (ii) nor (iii) satisfies the underlying requirement.} 
\end{example}

\bibliographystyle{IEEEtran}
\bibliography{PIR_salim,pir_bib}

\end{document}